\documentclass{ecai} 

\usepackage{latexsym}
\usepackage{amssymb}
\usepackage{amsmath}
\usepackage{amsthm}
\usepackage{booktabs}
\usepackage{enumitem}
\usepackage{graphicx}
\usepackage{color}

\newtheorem{definition}{Definition}
\newtheorem{example}{Example}

\newtheorem{observation}{Observation}

\newtheorem{theorem}{Theorem}
\newtheorem{property}{Property}
\newtheorem{corollary}{Corollary}

\newcommand{\BibTeX}{B\kern-.05em{\sc i\kern-.025em b}\kern-.08em\TeX}

\begin{document}


\begin{frontmatter}


\paperid{0107} 



\title{Weighted Envy-Freeness Revisited: Indivisible Resource and House Allocations}


\author[A]{\fnms{Yuxi}~\snm{Liu}}
\author[A]{\fnms{Mingyu}~\snm{Xiao}\orcid{0000-0002-1012-2373}\thanks{Corresponding Author. Email: myxiao@uestc.edu.cn}}


\address[A]{University of Electronic Science and Technology of China, Chengdu, China}


\begin{abstract}
Envy-Freeness is one of the most fundamental and important concepts in fair allocation. Some recent studies have focused on the concept of weighted envy-freeness.  Under this concept, each agent is assigned a weight, and their valuations are divided by their weights when assessing fairness. This concept can promote more fairness in some scenarios. But on the other hand, experimental research has shown that this weighted envy-freeness significantly reduces the likelihood of fair allocations. When we must allocate the resources, we may propose fairness concepts with lower requirements that are potentially more feasible to implement. In this paper, we revisit weighted envy-freeness and propose a new concept called SumAvg-envy-freeness, which substantially increases the existence of fair allocations. This new concept can be seen as a complement of the normal weighted envy-fairness. Furthermore, we systematically study the computational complexity of finding fair allocations under the old and new weighted fairness concepts in two types of classic problems: Indivisible Resource Allocation and House Allocation. Our study provides a comprehensive characterization of various properties of weighted envy-freeness.
\end{abstract}

\end{frontmatter}


\section{Introduction}

Given a set of valuable resources, the fair division problem asks whether these resources can be allocated among agents with potentially differing preferences in a fair manner.
This problem is important in economics and has garnered increasing attention in artificial intelligence and computer science over the past few decades~\citep{brams1996fair,markakis2017approximation,moulin2004fair,moulin2019fair,thomson2016fair}.
The problem has many applications, including land division~\citep{segal2017fair}, apartment rent sharing~\citep{gal2016fairest}, and divorce settlements~\citep{brams1996fair}.
Envy-freeness is one of the most widely studied fairness criteria in the literature.
It requires that each agent considers their assigned bundle to be at least as desirable as any other bundle in the allocation~\citep{foley1966resource,varian1973equity}.
For more information about fair allocations, we refer to two surveys by Amanatidis et al.~\citep{amanatidis2023fair} and Aziz et al.~\citep{aziz2022algorithmic}.

Recently, motivated by real-world applications where agents are often not equally obliged, Chakraborty et al.~\citep{chakraborty2021weighted} introduced the weighted setting.
In this framework, the weights assigned to agents can reflect widely recognized and accepted indicators of entitlement, such as eligibility or merit. A classic illustration of this is inheritance distribution, where individuals who are closer relatives typically have a greater claim to the inheritance than more distant relatives. Likewise, larger organizations with more individuals may be entitled to a larger share of resources. By incorporating weights, this model is able to account for a wider variety of real-world situations.
Weighted models have been extremely studied in fair allocations, and we will provide more information on weighted models later.

Under the most widely studied weighted envy-freeness, the utility of each agent is divided by his weight in the fairness. 
This setup is simple and can promote more fairness in some scenarios.
However, experimental tests show that this weighted envy-freeness significantly reduces the likelihood of fair allocations~\citep{chakraborty2021weighted}.
Our subsequent experiments will also show that the probability of weighted fair allocations is substantially less than the probability of fair allocations without weights.

When discussing fair allocations, we often introduce new fairness concepts because the fair allocations under the current concept may not exist.
In such cases, we must settle for a relaxed fairness concept.
For instance, envy-freeness allocation may not always exist, and we introduced new fairness concepts like MMS and EF1.
The weighted envy-freeness appears to take a different approach: it enhances fairness theoretically, yet diminishes the feasibility of equitable distribution in practice.
To understand this better, we reexamine weighted envy-freeness and propose a new envy-freeness concept under the weight setting, which is called \emph{SumAvg-envy-freeness}.
In fact, our envy-freeness integrates the classic envy-freeness and weighted envy-freeness concepts.
We will show allocations under the new envy-freeness concept is more practical and more ``likely'' to exist than that under the old weighted envy-freeness concept.

The allocation problems we consider mainly involve two scenarios. One is the allocation of indivisible resources, where each resource can only be assigned to one agent as a whole and we also require that all resources be allocated. Otherwise, not allocating any resources at all could be a trivial solution. The other one is the house allocation problem, which requires that each agent is assigned exactly one resource.
The house allocation problem is also an important problem in fairness allocations.
When considering the weighted setting in house allocation, the weights for a family can represent the number of members in that family.
Naturally, larger families require larger living spaces compared to smaller families.
Therefore, even if a smaller family receives a less valuable house than a larger family, there may still be no envy.

\subsection{More Related Work}
\paragraph{Weighted Models.}
Weighted models have been studied in a wide range of fair allocations under different concepts of fairness. In addition to weighted envy-freeness, there are other concepts of weighted fairness in the literature.
Chakraborty et al.~\citep{chakraborty2021weighted} introduced the concept of weighted envy-freeness up to one item (WEF1) for the allocation of goods and demonstrated that WEF1 allocations always exist and can be found in polynomial time.
Both Wu et al.~\citep{wu2023weighted} and Springer et al.~\citep{springer2024almost} established the existence and presented algorithms for the computation of WEF1 allocations for chores.
WPROP1 allocations have been shown to always exist for chores~\citep{branzei2024algorithms} and for a mixture of goods and chores~\citep{aziz2020polynomial}.
Li et al.~\citep{li2023truthful} established the existence and computation of WPROPX allocations for chores. 
The weighted version of MMS has also been studied for both goods~\citep{farhadi2019fair} and chores~\citep{aziz2019weighted}.
When we consider more general utility functions, Chakraborty et al.~\citep{chakraborty2021weighted} showed that WEF1 
allocations do not always exist for arbitrary monotonic utilities. There are also some researches combining weighted setting and non-additive utility functions~\citep{DBLP:conf/aaai/MontanariSST24,DBLP:conf/sigecom/ViswanathanZ23}.
For a comprehensive review of existing research on weighted indivisible fair allocation, please refer to the recent surveys~\citep{suksompong2024weighted}.

\paragraph{House Allocation.}
House allocation is a special case of indivisible resource allocation, where each agent gets exactly one resource.
This problem has been extensively studied in the context of designing  incentive-compatible mechanisms and ensuring economic efficiency~\citep{abdulkadirouglu2003school,svensson1999strategy}.
Recent research has increasingly focused on fairness, often defined through the notion of envy-freeness~\citep{varian1973equity,foley1966resource,aziz2024envy,beynier2019local}.
Gan et al.~\citep{gan2019envy} developed a polynomial-time algorithm to check and find an envy-free house allocation.
Recently,  Dai et al.~\citep{dai2024weighted} extended the weight setting to house allocation and also got a polynomial-time algorithm to check and find an weighted envy-free house allocation.
When the agents are previously partitioned into several groups, group-fairness was also extended to house allocation~\citep{DBLP:conf/ecai/Gross-HumbertBB23}.
In cases where fair allocations under existing fairness concepts may not be achievable, researchers have also explored compromises, aiming to balance fairness objectives. The problem of finding allocations that maximize multiple fairness criteria has been widely studied~\citep{kamiyama2021complexity,madathil2023complexity,hosseini2023graphical,DBLP:conf/atal/Hosseini0SVV24}.

\section{Our Model and Contributions}

For the sake of presentation, we call the envy-freeness without weight \emph{Sum-envy-freeness} and the previous weighted envy-freeness \emph{Avg-envy-freeness}.

\subsection{New Concept of Envy-Freeness}

As mentioned earlier, one motivation for proposing a new concept of fairness is the hope that fair allocations under this new concept will exist. In fact, most new fairness concepts typically exhibit the \textbf{\emph{inheritability}} property: if an allocation satisfies the original fairness concept, it will also satisfy the new one.

We also observe another property commonly held by previous fairness concepts: when an agent envies another, swapping the resources assigned to them can eliminate the envy. We refer to this property as \textbf{\emph{exchange elimination}}. This property is quite reasonable, and almost all existing fairness concepts satisfy it.

The concept of weighted envy-freeness (Avg-envy-freeness) does not satisfy either of the two properties mentioned above. In Example 1, when there are no weights, an envy-free allocation exists where each agent receives one resource. But after introducing weights, no envy-free allocation can be found. Thus, weighted envy-freeness does not satisfy the inheritability property.
Moreover, when each agent receives one resource, agent 2 envies agent 1 because agent 1's weight is larger. Even if we exchange the resources between them, agent 2 continues to envy agent 1. Therefore, the exchange elimination property does not hold.

\begin{example}
    There are two indivisible resources $r_1$ and $r_2$, and two agents $a_1$ and $a_2$ with weights $w_1 = 1$ and $w_2 = 2$.
    The utility of each agent on each resource is the same.
\end{example}

We hope that the two properties discussed above can be satisfied under the concept of weighted fairness. To achieve this, we hope that a fair allocation in the unweighted case should also be fair in the weighted case.

Based on this principle, we modify the concept of fairness in the weighted model as follows: agent A will not envy agent B if at least one of the following two conditions is met: (1) in agent A's view, the utility of the resources assigned to agent A divided by agent A's weight is no less than the utility of the resources assigned to agent B divided by agent B's weight (this corresponds to the condition for Avg-envy-freeness); (2) in agent A's view, the utility of the resources assigned to agent A is no less than the utility of the resources assigned to agent B (this corresponds to the condition for Sum-envy-freeness).
We refer to this new concept as \emph{SumAvg-envy-freeness}.

SumAvg-envy-freeness should not be viewed as a simple union of the previous concepts of Sum-envy-freeness and Avg-envy-freeness. For example, in Example 2, there exists a SumAvg-envy-free allocation where $r_1$ is assigned to $a_1$ and $r_2$ is assigned to $a_2$. However, there is no allocation that is Sum-envy-free or Avg-envy-free.

\begin{example}
    There are two indivisible resources $r_1$ and $r_2$, and two agents $a_1$ and $a_2$ with weights $w_1=1$ and $w_2= 10$.
    The utility of each agent on resource $r_1$ is the same 5 and the utility of each agent on resource $r_2$ is the same 10.
\end{example}

From the above examples, we can see that SumAvg-envy-free allocation is more likely to exist.
Our experiments further support this observation.
We test on 10000 weighted instances with 8 resources and 5 to 8 agents. The ratio of instances having fair allocations under different envy-free concepts is shown in Table~\ref{tab:preresult1}.
More discussion and experimental results in different settings are shown in Section 6.

It should be noted that we are not denying Avg-envy-freeness. 
In fact, our proposed concept of SumAvg-envy-freeness complements, rather than contradicts, the idea of Avg-envy-fairness.
SumAvg-envy-freeness exhibits the inheritability property not only for Sum-envy-freeness but also Avg-envy-freeness.
When Avg-envy-freeness does not exist, we still need to allocate the resources—how should we proceed? 
Perhaps SumAvg-envy-freeness can be used as a fallback concept to perform the allocation. 
Our proposal of SumAvg-envy-freeness is not a denial of Avg-envy-freeness, but rather an extension of the notion of fairness under weights. 
It's analogous to how EF1 was proposed when EF allocations were not guaranteed to exist. 

\begin{table}[!t]
    \centering
    \small
    \begin{tabular}{cccc}
        \hline
        $\#$ agents & Sum & Avg & SumAvg\\

        \hline
        5     & $19.63\%$ & $10.12\%$ & $98.02\%$\\
        \hline
        6     & $2.12\%$  & $0.52\%$  & $90.59\%$\\
        \hline
        7     & $0.45\%$  & $0.01\%$ & $69.81\%$\\
        \hline
        8     & $0.32\%$  & $<0.01\%$ & $27.90\%$\\
        \hline
    \end{tabular}
    \caption{The ratio of instances having fair allocations under different envy-free concepts among 10000 tested instances}
    \label{tab:preresult1}
\end{table}

\subsection{Computational Complexity}

After introducing a new concept of fairness, we need to study the existence of the fair allocation under this concept and the computational complexity of computing the fair allocation.
Although SumAvg-envy-freeness can be applied for a wider range, SumAvg-envy-free allocations may still not exist even in very simple scenarios. For example, in the case where it is to assign one resource to two agents, fair allocation will not exist. 
Next, we systematically investigate the computational complexity of finding a fair allocation under different fairness concepts.

We will first consider the NP-hardness of our problems. In fact, they will be computationally hard in general. Next, our research approach contains two ways.
One is to investigate whether a polynomial-time algorithm exists for the restricted version of the problem under different constraints.
The other is to study the parameterized complexity of the problem with different parameters.
We mainly consider two restricted versions on the utility functions which are \emph{identical} and \emph{0/1}.
We call an utility function \emph{identical} if the utility of each agent on any resource $r$ is the same.
We call an utility function 0/1 if the utility of each agent on any resource $r$ is either 0 or 1.
For parameterized complexity, we  will mainly consider two parameters: ``number of agents'' and ``number of resources''.
Previous results and our results under the three fair concepts Sum-envy-freeness, Avg-envy-freeness, and SumAvg-envy-freeness 
are presented in Table~\ref{tab:results}.

\begin{table*}[!t]
    \centering
    \small
    \begin{tabular}{cccc}
        \hline
        Preference Type & Sum & Avg & SumAvg\\
        \hline
        for $\#$agents & ~ & ~ & ~\\
        id. 0/1     & P~\citep{bredereck2022envy} & \textbf{P(Obs.~\ref{lemma:avg-poly})} & \textbf{P(Thm.~\ref{thm:P})}\\
        0/1         & FPT~\citep{bliem2016complexity}  & \textbf{FPT(Coro.~\ref{coro:avgfptn})} & \textbf{FPT(Coro.~\ref{coro:SumAvgfpt})} \\
        id. (unary) & W[1]-h~\citep{bliem2016complexity}  & W[1]-h~\citep{bliem2016complexity}  & W[1]-h~\citep{bliem2016complexity}  \\ 
        id. (binary)& para-NP-h~\citep{bouveret2008efficiency}  & para-NP-h~\citep{bouveret2008efficiency}  & para-NP-h~\citep{bouveret2008efficiency}  \\ 
        \hline
        for $\#$resources & ~ & ~ & ~\\
        add. mon. & FPT~\citep{bliem2016complexity} & \textbf{FPT(Coro.~\ref{coro:avgfptm})} & \textbf{FPT(Thm.~\ref{thm:fpt})} \\
        \hline
    \end{tabular}
    \caption{Parameterized complexity of EF-Allocation. The term ``add.mon.'' stands for ``additive monotonic.''. The term ``id.'' stands for ``identical''. The term ``unary'' means that the utility values are unary encodings while the term ``binary'' means that the utility values are binary encodings. Our results are in boldface.} 
    \label{tab:results}
\end{table*}

In Section 4, we consider the allocation of indivisible resources. We show that under identical and 0/1 preferences, checking the existence of Avg-envy-free and SumAvg-envy-free allocations can be solved in polynomial time. However, under general 0/1 preferences, both of these problems become NP-hard.

We then further demonstrate that under 0/1 preferences, checking the existence of Avg-envy-free and SumAvg-envy-free allocations is fixed-parameter tractable (FPT) with respect to the parameter ``number of agents." 
Under general preferences, we show that checking the existence of Avg-envy-free and SumAvg-envy-free allocations is FPT with respect to the parameter ``number of resources."

In Section 5, we shift focus to house allocation. 
Under general preferences, while checking the existence of Sum-envy-free and Avg-envy-free allocations can be solved in polynomial time, the problem of checking the existence of Sum-envy-free allocations becomes NP-hard, even when there are only two types of weights and three types of numbers in the utility functions. However, under 0/1 preferences or identical preferences, finding SumAvg-envy-free house allocations can be done in polynomial time.

\begin{table}[!t]
    \centering
    \begin{tabular}{cccc}
        \hline
        Preference Type & Sum & Avg & SumAvg\\
        \hline
        0/1         & P~\citep{gan2019envy}  & P~\citep{dai2024weighted}   & \textbf{P(Obs.~\ref{obs:house01P})} \\
        id.         & P~\citep{gan2019envy}  & P~\citep{dai2024weighted}   & \textbf{P(Thm.~\ref{thm:houseidP})} \\
        add. mon.   & P~\citep{gan2019envy}  & P~\citep{dai2024weighted}   & \textbf{NP-h(Thm.~\ref{thm:NP-h})}\\
        \hline
    \end{tabular}
    \caption{Classic complexity of EF-house-Allocation. The term ``add.mon.'' stands for ``additive monotonic.''. The term ``id.'' stands for ``identical''. Our results are in boldface.} 
    \label{tab:results2}
\end{table}

\section{Preliminaries}


An instance of the weighted fair allocation problem consists of a set $A = \{a_1, a_2, \dots, a_n\}$ of $n$ agents where each $a_i\in A$ has weight $w_i > 0$.
Let $R=\{r_1, r_2,\dots, r_m\}$ be a set of $m$ resources with utility functions $u_i : 2^R \rightarrow \mathbb{Z}$.
An \emph{allocation} of a set $R$ of indivisible resources to a set $A$ of agents is a mapping $\pi : A \rightarrow 2^R$ such that $\pi(a)$ and $\pi(a')$ are disjoint whenever $a \neq a'$. For any agent $a \in A$, we call $\pi(a)$ the \emph{bundle} of $a$ under $\pi$.
Furthermore, if for each $a_i\in A$, the size of $\pi(a_i)$ is exactly one, the allocation $\pi$ is called a \emph{house allocation}.

An utility function $u : 2^R \rightarrow \mathbb{Z}$ is \emph{additive} if for each bundle $X \subseteq R$, $u(X) = \sum_{r\in X} u(\{r\})$.
An additive utility function is \emph{monotonic} if it only outputs non-negative utilities.
In this paper, we assume that utility functions are additive and monotonic.

Next, we give the formal definitions of the three fairness concepts. 

\begin{definition}
For any pair $a_i$ and $a_j$ of agents in $A$,
    an allocation is Sum-envy-free (SEF) if it holds that
    \[
        u_i(\pi(a_i))\geq u_i(\pi(a_j));
    \]
An allocation is Avg-envy-free (AEF) if it holds that 
    \[
        \frac{u_i(\pi(a_i))}{w_i}\geq \frac{u_i(\pi(a_j))}{w_j};
    \]
An allocation is SumAvg-envy-free (SAEF) if it holds that
    \[
        u_i(\pi(a_i))\geq u_i(\pi(a_j)) \text{ or } \frac{u_i(\pi(a_i))}{w_i}\geq \frac{u_i(\pi(a_j))}{w_j}.
    \]

\end{definition}


\begin{definition}
    An allocation $\pi$ is complete if $\bigcup_{a\in A}\pi(a) = R$.
\end{definition}
When we consider indivisible resource allocations, we may always require the \emph{completeness} to avoid some trivial cases.

We define the following computational problems.

\noindent\rule{\linewidth}{0.2mm}
\textsc{SAEF-Allocation}\\ 
\textbf{Instance:} A set $A$ of $n$ agents where each $a \in A$ has weight $w_a > 0$, a set $R$ of $m$ indivisible resources, a family $U = \{u_1, u_2, \dots, u_n\}$ of non-negative utility functions.\\
\textbf{Task:}  To find a complete and SumAvg-envy-free allocation.\\
\rule{\linewidth}{0.2mm}

Similarly, we can define \textsc{SEF-Allocation} and \textsc{AEF-Allocation} by finding a Sum-envy-free allocation or an Avg-envy-free allocation instead of a SumAvg-envy-free allocation.
For house allocation, we define the following \textsc{SAEF-House-Allocation} problem.

\noindent\rule{\linewidth}{0.2mm}
\textsc{SAEF-House-Allocation}\\ 
\textbf{Instance:} A set $A$ of $n$ agents where each $a \in A$ has weight $w_a > 0$, a set $R$ of $m$ indivisible resources, a family $U = \{u_1, u_2, \dots, u_n\}$ of non-negative utility functions.\\
\textbf{Task:}  To find an SumAvg-envy-free house allocation.\\
\rule{\linewidth}{0.2mm}

Similarly, we can define \textsc{SEF-House-Allocation} and \textsc{AEF-House-Allocation}.




\section{Indivisible Resource Allocations}

In this section, we consider indivisible resource allocations. We will analyze the NP complexity and parameterized complexity for checking the existence of fairness allocations under the three envy-free concepts.
Recall that previous and our results are presented in Table \ref{tab:results}.

Consider an instance $(A, R, U)$ of \textsc{AEF-Allocation} or \textsc{SAEF-Allocation}, if for any agent $a\in A$, we have $w_a = 1$, this instance is equivalent to the same instance of \textsc{SEF-Allocation}.
Thus, the hardness results for \textsc{SEF-Allocation} will imply the same hardness results for \textsc{AEF-Allocation} and \textsc{SAEF-Allocation}.
The NP-hardness results under different restricted preferences of \textsc{SEF-Allocation} are established by Lipton et al.~\citep{lipton2004approximately}, Bouveret and Lang~\citep{bouveret2008efficiency}, and Aziz et al.~\citep{aziz2015fair}.

\subsection{Polynomial solvable cases}
We consider the case where the preferences are 0/1 and identical.
Firstly, we show \textsc{AEF-Allocation} under identical and 0/1 preferences can be solved in polynomial time by a simple observation.
Then we show a main result in this section that \textsc{SAEF-Allocation} under identical and 0/1 preferences can be solved in polynomial time.

\begin{observation}\label{lemma:avg-poly}
    \textsc{AEF-Allocation} under identical and 0/1 preferences can be solved in polynomial time.
\end{observation}

\begin{proof}
    Let $u$ be the identical utility function.
    For any Avg-envy-free allocation and any two agents $a_i, a_j$, we have that $u(\pi(a_i))/w_i \geq u(\pi(a_j))/w_j$ and $u(\pi(a_j))/w_j \geq u(\pi(a_i))/w_i$.
    Thus, we have that for any two agents $a_i, a_j$, $u(\pi(a_i))/w_i = u(\pi(a_j))/w_j$.
    So we can first calculate $m/\sum_{a\in A}w_a$ to represent the number of resources should be allocated per weight.
    Then, for each agent $a\in A$, we check whether $w_a(m/\sum_{a'\in A}w_{a'})$ is an integer to finish our algorithm.
\end{proof}

Next, we show that \textsc{SAEF-Allocation} under identical and 0/1 preferences can be solved in polynomial time.
Firstly, we prove the following two properties of SumAvg-envy-free allocations.

\begin{property}\label{property-1}
    Consider an instance $(A, R, U)$ of \textsc{SAEF-Allocation} where $A = \{a_1, a_2, \dots, a_n\}$ is sorted by weights in ascending order.
    Under identical preferences(let $u$ be the identical utility function), for any SumAvg-envy-free allocation $\pi$, we have that $u(\pi(a_i)) \leq u(\pi(a_{i+1}))$ for any $1 \leq i \leq n - 1$.
\end{property}

\begin{proof}
    By contradiction, we assume $u(\pi(a_i)) > u(\pi(a_{i+1}))$ for some $i$.
    Since $w_i\leq w_{i+1}$, we have that $u(\pi(a_i))/w_i > u(\pi(a_{i+1}))/w_{i + 1}$.
    In this case, $a_{i + 1}$ will envy $a_i$.
\end{proof}

\begin{property}\label{property-2}
    Consider an instance $(A, R, U)$ of \textsc{SAEF-Allocation} where $A = \{a_1, a_2, \dots, a_n\}$ is sorted by weights in ascending order.
    Under identical preferences(let $u$ be the identical utility function), for any SumAvg-envy-free allocation $\pi$, we have that $u(\pi(a_i))/w_i \geq u(\pi(a_{i+1})) / w_{i+1}$ for any $1 \leq i \leq n - 1$.
\end{property}

\begin{proof}
    By contradiction, we assume $u(\pi(a_i))/w_i < u(\pi(a_{i+1}))/w_{i+1}$ for some $i$.
    Since $w_i\leq w_{i+1}$, we have that $u(\pi(a_i)) < u(\pi(a_{i+1}))$.
    In this case, $a_{i}$ will envy $a_{i + 1}$.
\end{proof}


Consider an instance $(A, R, U)$ of \textsc{SAEF-Allocation} where $A = \{a_1, a_2, \dots, a_n\}$ is sorted by weights in ascending order.
Clearly, if an allocation $\pi$ satisfies that for any $1\leq i \leq n - 1$, $u(\pi(a_i)) \leq u(\pi(a_{i+1}))$ and $u(\pi(a_i))/w_i \geq u(\pi(a_{i+1})) / w_{i+1}$, then $\pi$ is a SumAvg-envy-free allocation.
By Property~\ref{property-1} and Property~\ref{property-2}, we have that $\pi$ is a SumAvg-envy-free allocation if and only if $\pi$ satisfies that for any $1\leq i \leq n - 1$, $u(\pi(a_i)) \leq u(\pi(a_{i+1}))$ and $u(\pi(a_i))/w_i \geq u(\pi(a_{i+1})) / w_{i+1}$.
Thus, in our algorithm, we search for an allocation satisfing that for any $1\leq i \leq n - 1$, $u(\pi(a_i)) \leq u(\pi(a_{i+1}))$ and $u(\pi(a_i))/w_i \geq u(\pi(a_{i+1})) / w_{i+1}$. We call such allocation \emph{feasible}.

Now we are ready to show our main algorithm.
\begin{theorem}\label{thm:P}
    \textsc{SAEF-Allocation} under identical and 0/1 preferences can be solved in $O(nm^3)$ time.
\end{theorem}

\begin{proof}
    Consider an instance $(A, R, U)$ of \textsc{SAEF-Allocation} where $A = \{a_1, a_2, \dots, a_n\}$ is sorted by weights in ascending order.
    For some $1\leq i\leq n, 1\leq j\leq m, 1\leq k\leq m$, we consider the following subproblem:
    to find a feasible allocation $\pi$ that allocates $j$ resources to agents $a_1, a_2, \dots, a_i$, where there are $k$ resources of the $j$ resources allocated to agent $a_i$. 
    We also let $c(i, j, k)$ denote the corresponding allocation.
    For some $1\leq i\leq n, 1\leq j\leq m, 1\leq k\leq m$, there may not exist any feasible allocation $\pi$, and we will let $c(i, j, k)=\emptyset$ for this case.
    To solve \textsc{SAEF-Allocation}, we only need to check the existence of the allocation among $c(i = n, j = m, k)$ for all possible $1\leq k\leq m$.
    Next, we use a dynamic programming method to compute all $c(i, j, k)$.

    For the case that $i = 1$, $j = k$, it trivially holds that $c(1, j, k)$ be the allocation that allocates $k$ resources to agent $a_1$.
    And for the case that $i = 1, j\neq k$, it trivially holds that $c(1, j, k) = \emptyset$.

    For every $2\leq i\leq n, 1\leq j\leq m, 1\leq k\leq m$,
    if there exists $k'\leq k$ such that $c(i - 1, j - k, k') \neq \emptyset$ and $k'/w_{i - 1} \geq k/w_i$ then
    $c(i, j, k)$ be the allocation that $c(i - 1, j - k, k')$ combined with allocating $k$ resources to agent $a_i$,
    otherwise $c(i, j, k) = \emptyset$.


    There are at most $nm^2$ different combinations of $(i, j, k)$.
    For each $2\leq i\leq n, 1\leq j\leq m$, and $1\leq k\leq m$,
    it takes at most $O(m)$ time to compute $c(i, j, k)$ by using the above recurrence relations.
    Therefore, our dynamic programming algorithm runs in $O(nm^3)$ time.
\end{proof}

\subsection{Few agents or few resources}

Although \textsc{SAEF-Allocation} under identical and 0/1 preferences can be solved in polynomial time,
Bouveret and Lang~\citep{bouveret2008efficiency} show that \textsc{SEF-Allocation} under identical preferences is NP-hard and \textsc{SEF-Allocation} under 0/1 preferences is also NP-hard.
Clearly, the hardness results for \textsc{SEF-Allocation} will imply the same hardness results for \textsc{AEF-Allocation} and \textsc{SAEF-Allocation}.
Thus, we turn to consider parameterized complexity for the hard problems.
We consider two parameters: the number $n$ of agents and the number $m$ of resources.
We first show that \textsc{SAEF-Allocation} is FPT with respect to the number of resources.

We encode our instance $(A, R, U)$ of \textsc{SAEF-Allocation} as an \textsc{Integer Programming} instance (SAEF-IP).
Let $A = \{a_1, a_2\dots, a_n\}$, $R = \{r_1, r_2, \dots, r_m\}$.
For some resource $r$, we define the \emph{type} of $r$ as a vector $t_r:= (u_1(r), u_2(r), \dots, u_n(r))$.
And let $T := \{t_r: r\in R\}$ be the set of types of all resources in $R$.
Let $t_r[i]$ be the value $u_i(r)$.
Let $\# t$ be the number of resources of type $t$.
Now we are ready to construct our ILP model.

For each agent $a_i\in A$ and each type $t\in T$,  we introduce a variable $x_i^t \in [m]$ to represent the number of resources of type $t$ allocated to agent $a_i$.
Since each feasible allocation is complete, we have the following constraint.
\begin{equation}
    \forall t\in T:\sum_{i\in[n]}x_i^t = \#t. %
\end{equation}
Since each feasible allocation is SumAvg-envy-free, we have the following two constraints, where at least one should be satisfied.
For each two agents $a_i, a_j\in A$,
\begin{equation}
    \sum_{t\in T}x_i^t t[i]\geq \sum_{t\in T}x_j^t t[i] %
\end{equation}
or
\begin{equation}
    w_j\sum_{t\in T}x_i^t t[i]\geq w_i\sum_{t\in T}x_j^t t[i]. %
\end{equation}

To represent this ``or'' constraint, we introduce a big number $M = \sum_{i\in [n], j\in [m]} u_i(j) \cdot \sum_{i\in [n]} w_i$.
Then, for each two agents $a_i, a_j\in A$, we introduce two variables $y_{ij}^1\in \{0, 1\}$ and $y_{ij}^2\in \{0, 1\}$ and introduce the following two constraints.
\begin{equation}
    M y_{ij}^1\leq M+(\sum_{t\in T}x_i^t t[i]- \sum_{t\in T}x_j^t t[i]), %
\end{equation}
\begin{equation}
    M y_{ij}^2\leq M+(w_j\sum_{t\in T}x_i^t t[i]- w_i\sum_{t\in T}x_j^t t[i]). %
\end{equation}
Consider the inequality (4), if $\sum_{t\in T}x_i^t t[i] \geq \sum_{t\in T}x_j^t t[i]$, we have that $M+(\sum_{t\in T}x_i^t t[i]- \sum_{t\in T}x_j^t t[i])\geq M$ and $y_{ij}^1$ can be 0 or 1.
If $\sum_{t\in T}x_i^t t[i] < \sum_{t\in T}x_j^t t[i]$, we have that $M+(\sum_{t\in T}x_i^t t[i]- \sum_{t\in T}x_j^t t[i])< M$ and $y_{ij}^1$ must be 0.
The similar arguments hold for $y_{ij}^2$.
Thus, we know that at least one of inequalities (2) and (3) is satisfied if and only if $y_{ij}^1 + y_{ij}^2 \geq 1$.
We have the following constraint.

\begin{equation}
    \forall i,j\in [n]: y_{ij}^1+y_{ij}^2\geq 1. %
\end{equation}

Now, we put all things together, and get the following \textsc{Integer Programming} instance (SAEF-IP).
\[
    \begin{split}
        & \mbox{Min }  1\\
        & \mbox{subject to:} \\
        & \forall t\in T:\sum_{i\in[n]}x_i^t = \#t\\
        & \forall i,j\in[n]: M y_{ij}^1<M+(\sum_{t\in T}x_i^t t[i]- \sum_{t\in T}x_j^t t[i])\\
        & \forall i,j\in[n]: M y_{ij}^2<M+(w_j\sum_{t\in T}x_i^t t[i]- w_i\sum_{t\in T}x_j^t t[i])\\
        & \forall i,j\in [n]: y_{ij}^1+y_{ij}^2\geq 1\\
        & \forall t\in T, i\in [n]: x_i^t \in [m]\\
        & \forall i,j\in [n]: y_{ij}^1\in \{0, 1\}\\
        & \forall i,j\in [n]: y_{ij}^2\in \{0, 1\}
    \end{split}
\]

\begin{theorem}\label{thm:fpt}
    \textsc{SAEF-Allocation} is fixed-parameter tractable with respect to the parameter ``number of resources''.
\end{theorem}

\begin{proof}
    Clearly, the number of variables in (SAEF-IP) is upper bounded by a function of the number $m$ of resources in an instance of \textsc{SAEF-Allocation}.
    The result is a consequence of applying the celebrated result of Lenstra~\citep{lenstra1983integer} for ILP models with a bounded number of variables.
\end{proof}

As a corollary, we show that the FPT result also holds for the case of 0/1 preferences.


\begin{corollary}\label{coro:SumAvgfpt}
    \textsc{SAEF-Allocation} under 0/1 preferences is fixed-parameter tractable with respect to the parameter ``number of agents''.
\end{corollary}

\begin{proof}
    Note that under 0/1 preferences, there are at most $2^n$ different resources types.
    Thus, the number of variables in (SAEF-IP) is upper bounded by a function of the number $n$ of resources in an instance of \textsc{SAEF-Allocation}.
    The result is a consequence of applying the celebrated result of Lenstra~\citep{lenstra1983integer} for ILP models with a bounded number of variables.
\end{proof}

For \textsc{AEF-Allocation},
we encode our instance $(A, R, U)$ as an \textsc{Integer Programming} instance (AEF-IP) by a similar way.
\[
    \begin{split}
        \mbox{Min } &  1\\
        \mbox{subject to: } & \forall t\in T:\sum_{i\in[n]}x_i^t = \#t\\
        & \forall w_j\sum_{t\in T}x_i^t t[i]\geq w_i\sum_{t\in T}x_j^t t[i]\\
        & \forall t\in T, i\in [n]: x_i^t \in [m]
    \end{split}
\]
And by similar arguments, we give the following corollaries without proofs.

\begin{corollary}\label{coro:avgfptm}
    \textsc{AEF-Allocation} is fixed-parameter tractable with respect to the parameter ``number of resources''.
\end{corollary}

\begin{corollary}\label{coro:avgfptn}
    \textsc{AEF-Allocation} under 0/1 preferences is fixed-parameter tractable with respect to the parameter ``number of agents''.
\end{corollary}

\section{House Allocations}
In this section, we consider house allocations. Recall that previous and our results are presented in Table \ref{tab:results2}.



\textsc{SEF-house-Allocation} and \textsc{AEF-house-Allocation} under additive monotonic preferences can be solved in polynomial time~\citep{gan2019envy,dai2024weighted}.
Surprisingly, we demonstrate that \textsc{SAEF-house-Allocation} under additive monotonic preferences is NP-hard, even when there are only two types of weights.
The hardness result is obtained by reducing from the classic NP-complete problem \textsc{3-SAT}~\citep{hartmanis1982computers}.

\noindent\rule{\linewidth}{0.2mm}
\textsc{3-SAT}\\ 
\textbf{Instance:} a set of clauses $C = \{c_1, \dots, c_m\}$ defined over a set of variables $X = \{x_1, \dots, x_n\}$ such that each clause is disjunctive and consists of 3 literals.\\
\textbf{Task:}  Determine whether there exists an assignment of the variables which satisfies all the clauses.\\
\rule{\linewidth}{0.2mm}


\begin{theorem}\label{thm:NP-h}
    \textsc{SAEF-House-Allocation} under additive monotonic preferences is NP-hard even when there are only two types of weights and three types of numbers in utility functions.
\end{theorem}

\begin{proof}
    We will show a polynomial-time reductions from \textsc{3-SAT} to \textsc{SAEF-House-Allocation} under additive monotonic preferences.
    Specifically, consider a \textsc{3-SAT} instance $(C = \{c_1, \dots, c_m\}, X = \{x_1, \dots, x_n\})$, we construct an equivalent \textsc{SAEF-House-Allocation} instance $(A = A_x\cup A_c, R = R_x\cup R_c, U = U_x\cup U_c)$ as follows.


    Let $M$ be a large constant number. We ensure that the maximum number in utility functions is $M$.

        \textit{(a) Variable gadgets:}
        For each variable $x_i\in X$, we construct two agents $a_{x, i}\in A_x$ and $\bar a_{x, i}\in A_x$, and two resources $r_{x, i}\in R_x$ and $\bar r_{x, i}\in R_x$.
        Let $\overline{\overline{r}}_{x, i} =  r_{x, i}$.
        Let
        \[
            u_{a_{x, i}}(r_{x, i}) = u_{a_{x, i}}(\bar r_{x, i}) = 1.
        \]
        and
        \[
        u_{\bar a_{x, i}}(r_{x, i}) = u_{\bar a_{x, i}}(\bar r_{x, i}) = M.
        \]
        And for any other possible resource $r$, let
        \[
            u_{a_{x, i}}(r) = u_{\bar a_{x, i}}(r) = 0.
        \]
        Let $w_{a_{x, i}} = 1$ and $w_{\bar a_{x, i}} = M$.
        Clearly, consider any SumAvg-free-allocation $\pi$, one of $r_{x, i}$ and $\bar r_{x, i}$ will be allocated to $a_{x, i}$ and the other will be allocated to $\bar a_{x, i}$.
        If agent $a_{x, i}$ is allocated resource $r_{x, i}$, it can be interpreted in 3-SAT as setting variable $x_i$ to true.
        Similarly, if agent $a_{x, i}$ is allocated resource $\bar r_{x, i}$, it can be interpreted in 3-SAT as setting variable $x_i$ to false.

        \textit{(b) Clause gadgets:} 
        For each clause $c_j$, we construct four agents $a_{c, j}^1, a_{c, j}^2, a_{c, j}^3$ and $a_{c, j}^*\in A_c$. And we construct four resources $r_{c, j}^1, r_{c, j}^2, r_{c, j}^3$ and $r_{c, j}^*\in R_c$.
        Let $c_j = l(j, 1)\vee l(j, 2)\vee l(j, 3)$, where $l(j, k)$ respects the $k$-th literal in $c_j$.
        We use $r_{l(j, k)}$ to denote the resource corresponding to the $k$-th literal in $c_j$.
        For example, if $c_1 = x_1\vee x_2 \vee \bar x_n$, then $l(1, 3) = \bar x_n$, $r_{l(1, 3)} = \bar r_{x, n}$ and $\bar r_{l(1, 3)} = r_{x, n}$.

        Consider the utility functions for $a_{c, j}^k$ ($k = 1, 2, 3$), we construct them as follows.
        Let
        \[
            u_{a_{c, j}^k}(r_{c, j}^k) = u_{a_{c, j}^k}(\bar r_{l(j, k)}) = M,
        \]
        and\\
        \[
            u_{a_{c, j}^k}(r_{c, j}^*) = 1.
        \]
         
         For any other possible resource $r$, let
        $u_{a_{c, j}^k}(r) = 0$.
        
        For their weights, let $w_{a_{c, j}^1} = w_{a_{c, j}^2} = w_{a_{c, j}^3} = 1$.

        Consider the utility functions for $a_{c, j}^*$, let
        \[
            u_{a_{c, j}^*}(r_{c, j}^1) = u_{a_{c, j}^*}(r_{c, j}^2) = u_{a_{c, j}^*}(r_{c, j}^3) = M.
        \]
        
        For any other possible resource $r$, let $u_{a_{c, j}^*}(r) = 0$.
        For the weight, let $w_{a_{c, j}^*} = M$.

        \begin{figure}[!t]
            \centering
            \includegraphics[width = 6cm]{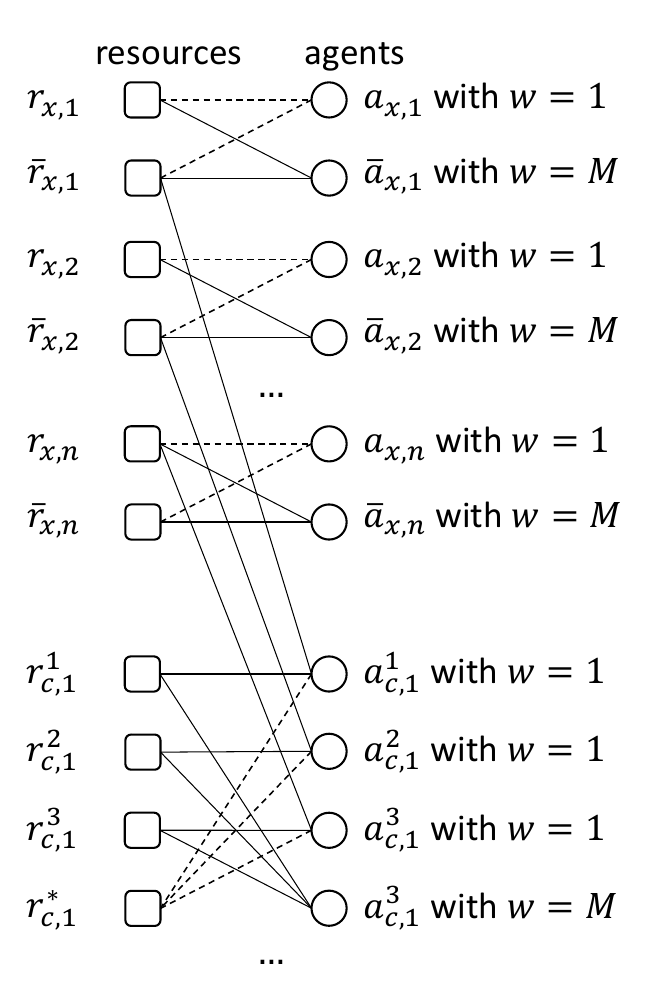}
            \caption{
                An illustration for $(A, R, U)$, where $r_1 = x_1\vee x_2\vee \bar x_n$.
                For each agent $a$ and resource $r$, (i) there is a solid edge between $a$ and $r$ if and only if $u_a(r) = M$; (ii) there is a dashed edge between $a$ and $r$ if and only if $u_a(r) = 1$; (ii) there is no edge between $a$ and $r$ if and only if $u_a(r) = 0$.
            }
            \label{pic:reduction}
        \end{figure}

    We finish our construction of the \textsc{SAEF-House-Allocation} instance.
    See Fig.~\ref{pic:reduction} for an illustration.
    Now, we show that $(A, R, U)$ is a yes-instance of \textsc{SAEF-House-Allocation} if and only if $(C, X)$ is a yes-instance of \textsc{3-SAT}.
    
    Assume $(C, X)$ is a yes-instance, we construct a SumAvg-envy-free house allocation $\pi$ as follows.
    For each variable $x_i$, if $x_i$ is true, then $\pi(a_{x,i}) = r_{x, i}$ and $\bar \pi(\bar a_{x,i}) = \bar r_{x, i}$.
    Otherwise, $\pi(a_{x,i}) = \bar r_{x, i}$ and $\pi(\bar a_{x,i}) = r_{x, i}$.
    For each clause $c_j = l(j, 1)\vee l(j, 2)\vee l(j, 3)$, if $l(j, 1)$ is true, then $\pi(a_{c, j}^1) = r_{c, j}^*$.
    Otherwise, $\pi(a_{c, j}^1) = r_{c, j}^1$.
    Similarly, if $l(j, 2)$ (resp. $l(j, 3)$) is true and $r_{c, j}^*$ is not be allocated, then $\pi(a_{c, j}^2)$ (resp. $\pi(a_{c, j}^3)$) $= r_{c, j}^*$.
    otherwise, $\pi(a_{c, j}^2) = r_{c, j}^2$ (resp. $\pi(a_{c, j}^3) = r_{c, j}^3$).
    Since $(C, X)$ is a yes-instance, we know that at least one literal in $l(j, 1), l(j, 2)$ and $l(j, 3)$ should be true.
    Let $l(j, k)$ ($k = 1, 2, 3$) be the first literal that be true, which means that $\pi(a_{c, j}^k) = r_{c, j}^*$, then $\pi(a_{c, j}^*) = r_{c, j}^k$.

    We show that $\pi$ is a SumAvg-envy-free allocation as follows.
    Firstly, for some agent $a_i$ allocated with a resource with utility $M$, $a_i$ will not envy any other agent, since $a_i$ is allocated with the largest utility in all possible utilities.
    Thus, the only possible envy will occurs in the agent $a_{c, j}^k$ allocated with $r_{c, j}^*$ for some $j$ and $k$.
    Note that when $a_{c, j}^k$ is allocated with $r_{c, j}^*$,
    we know that $r_{c, j}^k$ is allocated to $a_{c, j}^*$ and $\bar r_{l(j, k)}$ is allocated to $\bar a_{x, i}$ since $l(j, k)$ is true.
    Note that the weights of agent $a_{c, j}^*$ and agent $\bar a_{x, i}$ are both $M$.
    We have that agent $a_{c, j}^k$ will not envy agent $a_{c, j}^*$ and agent $\bar a_{x, i}$ since $1 / 1 \geq M / M$.
    Thus, $\pi$ is a SumAvg-envy-free allocation.

    Assume $(A, R, U)$ is a yes-instance and let $\pi$ be a SumAvg-envy-free house allocation. We show that $(C, X)$ is a yes-instance.
    We construct the truth assignment as follows.
    Clearly, one of $r_{x, i}$ and $\bar r_{x, i}$ will be allocated to $a_{x, i}$ and the other will be allocated to $\bar a_{x, i}$.
    If agent $a_{x, i}$ is allocated resource $r_{x, i}$, we set $x_i$ to true, otherwise we set $x_i$ to false.
    By contradiction, we assume there exists a clause $c_j = l(j, 1)\vee l(j, 2)\vee l(j, 3)$ such that $l(j, 1) = l(j, 2) = l(j, 3) =$ false.
    In this case, we know that the corresponding resources of $\bar l(j, 1)$, $\bar l(j, 2)$ and $\bar l(j, 3)$ are allocated to agents with weights 1.
    Since there must be a $k = 1, 2, 3$ such that $r_{c, j}^k$ is allocated to $a_{c, j}^*$, there must be an agent $a_{c, j}^k$ allocated with $r_{c, j}^*$.
    However, $a_{c, j}^k$ will envy the agent allocated with $\bar r_{l(j, k)}$ since the weights of these two agents are both 1 and $a_{c, j}^k$ perfer $\bar r_{l(j, k)}$ than $a_{c, j}^*$, which leads a contradiction.

    Thus, this theorem holds.
\end{proof}

Now we show that under more restricted preferences, \textsc{SAEF-house-Allocation} can be solved in polynomial time.
Under 0/1 preferences, we have the following observation.

\begin{observation}\label{obs:house01P}
    \textsc{SAEF-House-Allocation} under 0/1 preferences can be solved in polynomial time.
\end{observation}

\begin{proof} 
    Consider any allocation $\pi$.
    Since each agent is allocated exactly one resource, for any pair $a_i$ and $a_j$ of agents, we know that 
    $\frac{u_i(\pi(a_i))}{w_i}\geq \frac{u_i(\pi(a_j))}{w_j}$ if and only if $ u_i(\pi(a_i))\geq u_i(\pi(a_j))$.
    In this case, this problem is equivalent to \textsc{SEF-house-Allocation}.
    Since \textsc{SEF-house-Allocation} under 0/1 preferences can be solved in polynomial time~\citep{gan2019envy}, we have that \textsc{SAEF-House-Allocation} under 0/1 preferences can be solved in polynomial time.
\end{proof}

Under identical preferences, we can design a dynamic programming algorithm for \textsc{SAEF-house-Allocation}, which is similar to the algorithm given in Theorem \ref{thm:P}.

\begin{theorem}\label{thm:houseidP}
    \textsc{SAEF-house-Allocation} under identical preferences can be solved in $O(nm^2)$ time. 
\end{theorem}

\begin{proof}
    Consider an instance $(A, R, U)$ of \textsc{SAEF-house-Allocation} where $A = \{a_1, a_2, \dots, a_n\}$ is sorted by weights in ascending order and $R = \{r_1, r_2, \dots, r_m\}$ is sorted by utilities in $u$ in ascending order.
    Firstly, by similar arguments, it is not hard to see Property~\ref{property-1} and Property~\ref{property-2} still hold for \textsc{SAEF-house-Allocation} under identical preferences. 
    Thus, our algorithm search for an allocation satisfing that for any $1\leq i \leq n - 1$, $u(\pi(a_i)) \leq u(\pi(a_{i+1}))$ and $u(\pi(a_i))/w_i \geq u(\pi(a_{i+1})) / w_{i+1}$.
    We still call such allocation \emph{feasible}. 

    For some $1\leq i\leq n, 1\leq j\leq m$, we consider the following subproblem:
    to find a feasible allocation $\pi$ that allocate resources from first $j$ resources to agents $a_1, a_2, \dots, a_i$ and allocate resource $r_j$ to agent $a_i$.
    We also let $c(i, j)$ denote the corresponding allocation.
    For some $1\leq i\leq n, 1\leq j\leq m$, there may not exist any feasible allocation $\pi$, and we will let $c(i, j)=\emptyset$ for this case.
    To solve \textsc{SAEF-house-Allocation}, we only need to check the existence of the allocation among $c(i = n, j)$ for all possible $n\leq j\leq m$.
    Next, we use a dynamic programming method to compute all $c(i, j)$.

    For the case that $i = 1$, for any $1\leq j\leq m$, it trivially holds that $c(1, j)$ be the allocation that allocates resource $r_j$ to agent $a_1$.

    For every $2\leq i\leq n, 1\leq j\leq m$,
    if there exists $j'\leq j$ such that $c(i - 1, j') \neq \emptyset$ and $u(r_{j'})/w_{i - 1} \geq u(r_j)/w_i$ then
    $c(i, j)$ be the allocation that $c(i - 1, j')$ combined with allocating resource $r_j$ to agent $a_i$,
    otherwise $c(i, j) = \emptyset$.

    There are at most $nm$ different combinations of $(i, j)$.
    For each $2\leq i\leq n, 1\leq j\leq m$,
    it takes at most $O(m)$ time to compute $c(i, j)$ by using the above recurrence relations.
    Therefore, our dynamic programming algorithm runs in $O(nm^2)$ time.
    Thus, this theorem holds.
\end{proof}

\section{Experiments}

To better understand the distinctions among SumAvg-envy-freeness, Sum-envy-freeness and Avg-envy-freeness,
it is important to investigate the existence of them in practical.
In this section, we address this question through a series of experiments.

We run 10,000 instances with $5, 6, 7, 8$ agents and 8 resources.
The linear preferences of the agents are generated either from \emph{impartial culture} (IC), with no restriction of domain, or following a distribution for preferences restricted to the \emph{single-peaked} domain~\cite{black1948rationale}.
Let us recall that a preference order $\succ$ is single-peaked with respect to an axis $>^O$ over the objects if there exists a unique peak object $x^* \in O$ such that for every pair of objects $a$ and $b$, $x^* >^O a >^O b$ implies $x^*\succ a \succ b$, and $a >^O b >^O x^*$ implies that $x^*\succ  b \succ a$.

In our experiments, the single-peaked preferences are generated from the \emph{single-peaked uniform peak} culture (SPUP), meaning they are generated by uniformly drawing a peak alternative on a given axis over the objects and then iteratively choosing the next preferred alternatives with equal probability on either the left or right of the peak along the axis.
The values of utility functions are drawn independently from a uniform distribution ranging from $1$ to $10,000$.
The weights of agents are drawn independently from two uniform distributions ranging from $1$ to $100$, and from $101$ to $200$, respectively.


The frequency of existence of an EF allocation in three different concepts are shown in Fig. \ref{Fig:2}.
The frequency of existence of an EF house allocation three different concepts are shown in Fig. \ref{Fig:3}.

\begin{figure}[!t]
    \centering
    \includegraphics[width = 8.5cm]{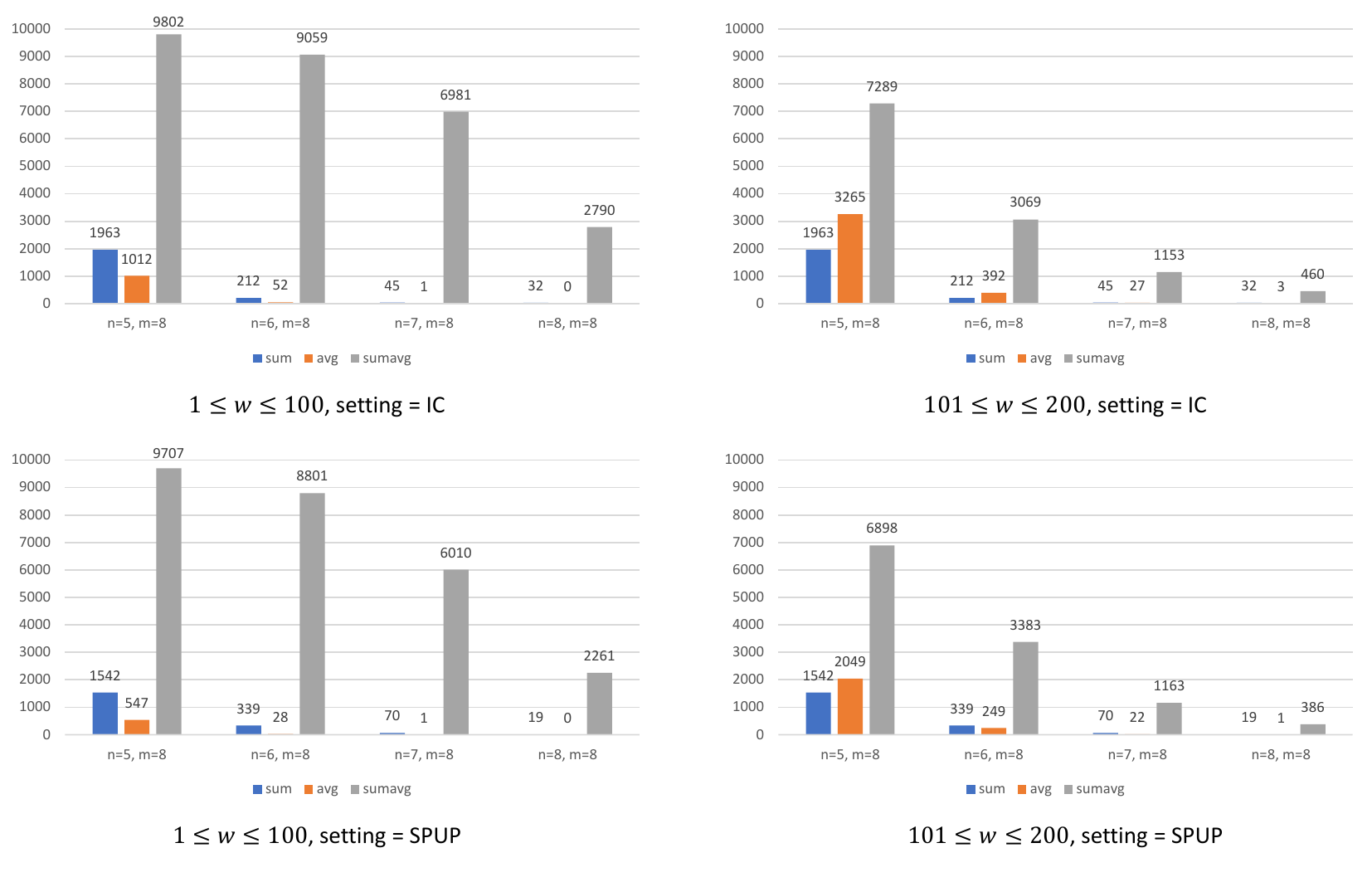}
    \caption{
    The number of instances that admit SumAvg-envy-free allocations(resp. Sum-envy-free allocations or Avg-envy-free allocations) in four different settings.
    The term $1\leq w\leq 100$ means that the weights of agents are drawn uniformly and independently from 1 to 100, while $101\leq w \leq 200$ means that the weights of agents are drawn uniformly and independently from 101 to 200.
    The term ``IC'' means the linear preferences of the agents are generated from impartial culture and ``SPUP'' means the linear preferences of the agents are generated from the single-peaked uniform peak culture.}
    \label{Fig:2}
\end{figure}

\begin{figure}[!t]
    \centering
    \includegraphics[width = 8.5cm]{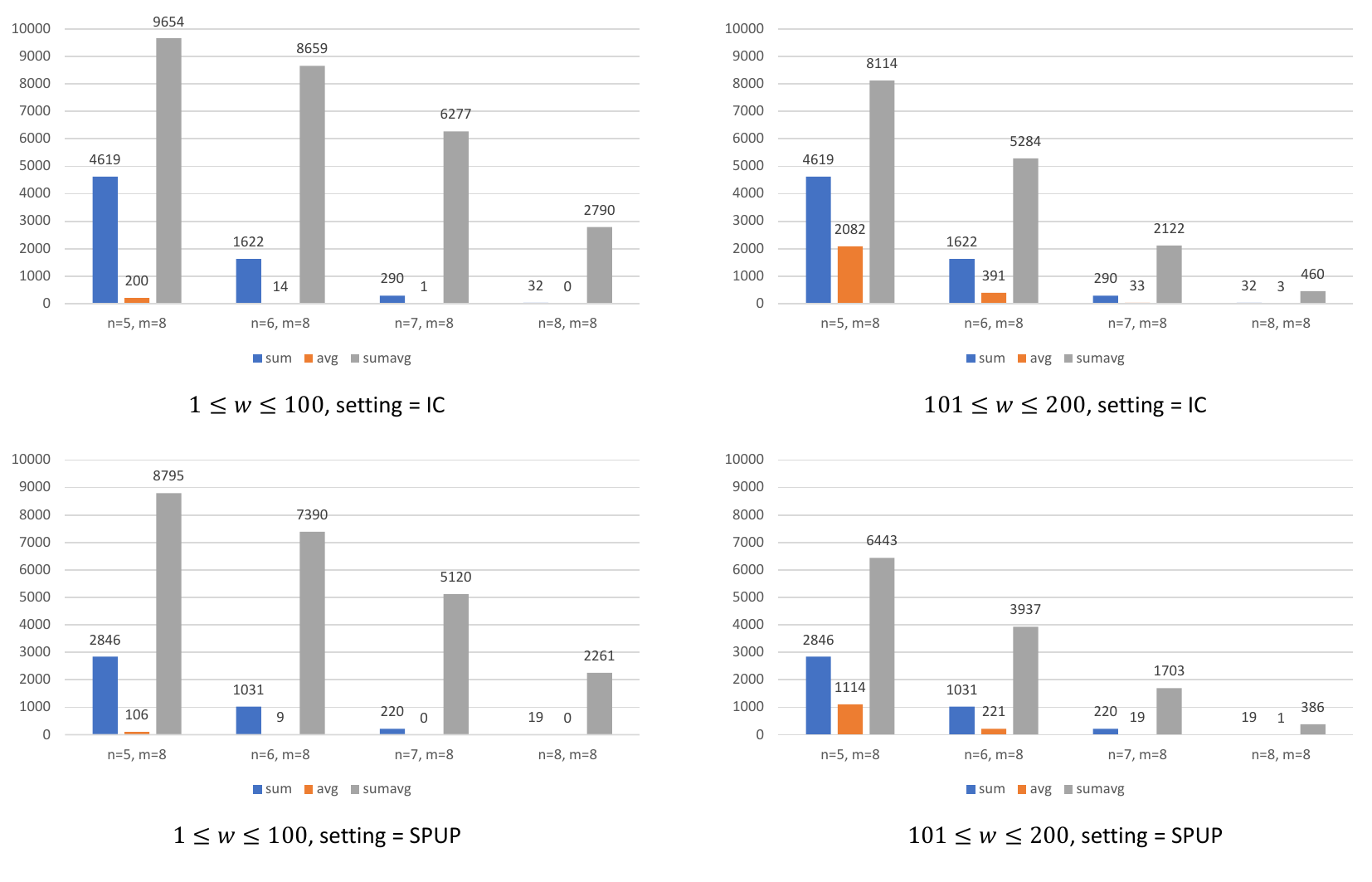}
    \caption{
    The number of instances that admit SumAvg-envy-free allocations(resp. Sum-envy-free allocations or Avg-envy-free allocations) in four different settings.
    The term $1\leq w\leq 100$ means that the weights of agents are drawn uniformly and independently from 1 to 100, while $101\leq w \leq 200$ means that the weights of agents are drawn uniformly and independently from 101 to 200.
    The term ``IC'' means the linear preferences of the agents are generated from impartial culture and ``SPUP'' means the linear preferences of the agents are generated from the single-peaked uniform peak culture.}
    \label{Fig:3}
\end{figure}

The experimental results reveal several key observations:


\begin{enumerate}
    \item In every setting, the number of instances admitting SumAvg-envy-free allocations is surprisingly larger than the number of instances admitting Sum-envy-free or Avg-envy-free allocations.
    \item In almost every setting, the number of instances admitting Sum-envy-free allocations is larger than those admitting Avg-envy-free allocations, aligning with the experimental results in~\citep{chakraborty2021weighted}.
    \item Compared to the setting where $101 \leq w \leq 200$, under the  $1\leq w\leq 100$ setting, the larger ratio between maximum and minimum weights exerts a stronger influence on weighted allocations.
We can see that in the $1\leq w\leq 100$ setting, SumAvg-envy-free allocations are more likely to exist while Avg-envy-free allocations are more likely to not exist. This outcome is consistent with expectations.
\end{enumerate}

Our experimental results demonstrate that achieving Sum-envy-freeness or Avg-envy-freeness is significantly more challenging than achieving SumAvg-envy-freeness.
When Sum-envy-free or Avg-envy-free allocations do not exist,
SumAvg-envy-free allocations may be able to serve as a viable alternative.

\section{Conclusion}

In this paper, we revisit the concept of weighted envy-freeness. 
To ensure the existence of fair allocations in broader scenarios, we introduce a new weighted fairness concept. While this approach may relax certain fairness guarantees, it significantly enhances allocation feasibility. Subsequently, we conduct a systematic computational complexity analysis of computing fair allocations under different fairness concepts.

We conclude with an interesting open problem. We have shown that \textsc{SAEF-House-Allocation} under additive monotonic preferences is NP-hard, while \textsc{SAEF-House-Allocation} under identical preferences can be solved in $O(nm^2)$ time. However, it remains unclear whether \textsc{SAEF-House-Allocation} under \emph{identical order preferences} can be solved in polynomial time. Identical order preferences mean that there exists an ordering of the resources $r_1, r_2, \dots, r_m$ such that $u_i(r_1)\geq u_i(r_2)\geq \dots \geq u_i(r_m)$ for every agent $a_i$.









\bibliography{WA}

\end{document}